\documentclass[11pt]{amsart}
\usepackage[margin=1.25in]{geometry}
\usepackage{amsfonts, float, mathtools}
\usepackage{amssymb}
\usepackage{euscript}
\usepackage{color}
\usepackage{cite}

\usepackage[pdftex]{hyperref}

 \newtheorem{thm}{Theorem}[section]

 \theoremstyle{definition}
 \newtheorem{defn}[thm]{Definition}
  
 \theoremstyle{remark}

 \numberwithin{equation}{section}

 \def\idty{{\mathchoice {\mathrm{1\mskip-4mu l}} {\mathrm{1\mskip-4mu l}} %
{\mathrm{1\mskip-4.5mu l}} {\mathrm{1\mskip-5mu l}}}}

\newcommand{\bR}{{\mathbb R}}

\newcommand{\Ir}{{\mathbb Z}}

\newcommand{\cA}{{\mathcal A}}
\newcommand{\cB}{{\mathcal B}}

\newcommand{\supp}{\operatorname{supp}}

\newcommand{\cS}{{\mathcal S}}

\newcommand{\A}{\mathcal{A}} 
\newcommand{\dom}{\mathop{\rm dom}}

\newcommand{\be}{\begin{equation}}
\newcommand{\ee}{\end{equation}}
\newcommand{\bea}{\begin{eqnarray}}
\newcommand{\eea}{\end{eqnarray}}
\newcommand{\beann}{\begin{eqnarray*}}
\newcommand{\eeann}{\end{eqnarray*}}

\newcommand{\Rl}{\bR}

\newcommand{\eq}[1]{(\ref{#1})}

\title[Lieb-Robinson Bounds]{From Lieb-Robinson Bounds to Automorphic Equivalence}

\author[B. Nachtergaele]{Bruno Nachtergaele}
\thanks{Based upon work supported by the National Science Foundation under DMS--2108390. MSC: 82B10 (Primary), 82B20, 82C10. Keywords: Lieb-Robinson bounds, quantum phase transitions, gapped ground states.}
\address{Department of Mathematics and Center for Quantum Mathematics and Physics\\
University of California, Davis\\
Davis, CA 95616, USA}
\email{bxn@math.ucdavis.edu}

\begin{document}
\date{\today }
\begin{abstract}
I review the role of Lieb-Robinson bounds in characterizing and utilizing the locality properties of the Heisenberg dynamics of quantum lattice systems. In particular, I discuss two definitions of gapped ground state phases and show that they are essentially equivalent.
\end{abstract}

\maketitle

\begin{center}
{\em Dedicated to Elliott Lieb on the occasion of his 90th birthday.}
\end{center}

\section{Introduction}\label{sec:intro}

Lieb and Robinson proved their now famous and ubiquitous propagation estimate in a 1972 article in Communications in Mathematical Physics \cite{lieb:1972}. The result immediately attracted attention. It clearly established an important fundamental property of the dynamics of quantum spin systems when the interactions are sufficiently short-range. With this 7-page paper they planted the seed of a revolution. That seed lay dormant for over thirty years. It took that long before anyone understood that Lieb-Robinson bounds are the key to proving a number of important results. It was Matthew Hastings who realized that a propagation estimate of the Lieb-Robinson type should hold and that it could be used, for example, to prove exponential clustering in a gapped ground 
state. He applied Lieb-Robinson bounds without knowing, at first, that such an estimate already existed in the literature and without knowing its proof. The estimate of a finite correlation length in terms of the ground state gap played an essential role in Hastings' breakthrough paper on the multi-dimensional Lieb-Schultz-Mattis theorem \cite{hastings:2004}. An exponential clustering theorem had been proved for a while in the context of relativistic quantum field theories \cite{fredenhagen:1985}, where the speed of light provides an absolute bound on the speed of propagation of signals. Rigorous proofs of the non-relativistic exponential clustering theorem and the Lieb-Robinson bounds used in them appeared in \cite{nachtergaele:2006a,hastings:2006}.

Interest in deriving Lieb-Robinson bounds for new situations continues unabated. The increasing variety of physical and `artificial' quantum many-body systems available for the experimentalist is the primary motivation for this, but there are also theoretical considerations that prompt researchers to try to improve and extend existing bounds. For the mathematical physicist this continuing interest is both a blessing and a curse. On the one hand, new questions prompt new interesting research but, on the other hand, the `need for speed' in some communities leads to a sometimes premature rush to publication. Some potentially interesting results remain unverified to my knowledge. In part for this reason, I will not focus on the latest improvements and extensions that make Lieb-Robinson type bounds available for an ever wider variety of physical systems. Instead, I focus on the role Lieb-Robinson bounds play in the definition and classification of gapped ground state phases. That topic too continues to develop and we restrict ourselves here to reviewing the importance of locality properties and the fundamental role Lieb-Robinson bounds play in expressing the locality properties quantitatively.

\section{Lieb-Robinson bounds}

Lieb-Robinson bounds are a quantitative expression of the (quasi-)locality of the dynamics of a quantum many-body system. This locality property stems from the decay of the interaction strength as a function of distance. There are different ways of imposing  sufficient decay conditions on the interactions that allow one to prove a Lieb-Robinson bound. The form the bounds take depends, among other things, on the form of the assumptions. It is not my goal here to review the many options that have been discussed in the literature, but rather to emphasize how the concept of locality is playing a central role in current research on quantum many-body physics.

\subsection{Locality of observables}

For simplicity of the presentation, we will restrict ourselves to quantum spin systems.  It is useful to leave the `lattice' as general as possible. This means considering spins associated with the points in a discrete metric space $(\Gamma, d)$ that obeys a growth condition on the number of sites in balls: we will assume that there are constants $c, \nu>0$ such that the number of points in a ball of radius $r$ is bounded by $1 + c r^\nu$. Some authors assume that the set is Delone, meaning that $\Gamma$ is embedded in $\Rl^d$, for some $d$, that the bound on the growth of balls follows from a minimum distance between points and, in addition, that there is also a radius $r_0>0$, such that $|B_{x}(r)|\geq 2$ for all $r>r_0, x\in \Gamma$. 
Such an assumption is not needed for the main results I discuss here. As a concrete example of $(\Gamma,d)$, one can think of $\Ir^\nu$, 
$\nu\geq 1$, with the usual $\ell^1$ distance.

For each $x\in\Gamma$, the observable algebra is a finite-dimensional complex matrix algebra $\A_{\{x\}}$, and for finite 
$\Lambda\subset \Gamma$,
$$
\A_\Lambda = \bigotimes_{x\in \Lambda} \A_{\{x\}}.
$$
For $\Lambda_1\subset \Lambda_2$, $\A_{\Lambda_1}$ is naturally embedded into $\A_{\Lambda_2}$ and therefore we can define
$$
\cA_{\rm loc} = \bigcup_{\mbox{\small finite }\Lambda\subset\Gamma}
\cA_\Lambda, \quad \cA_\Gamma =  \overline{\cA_{\rm loc}}^{\Vert\cdot\Vert},
$$
where the completion is taken with respect to the standard operator norm. This turns $\cA_\Gamma$ into a $C^*$-algebra.
The algebras $\cA^{\rm loc}_\Gamma$ and $\cA_\Gamma$  are referred to as the algebras of {\em local} and {\em 
quasi-local} observables, respectively. $A\in\A_\Lambda$ is said to be {\em supported in} $\Lambda$, and the {\em support} of $A$, $\supp A$, is the smallest $\Lambda$ for which $A\in\A_\Lambda$.

By construction, for all $A\in \cA_\Gamma$ and any increasing sequence of finite $\Lambda_n \uparrow \Gamma$,
there exists a sequence of local observables $A_n \in \cA_{\Lambda_n}$, such that  $A_n\to A$. We will characterize 
the rate of convergence by positive non-increasing functions $f$ with $\lim_n f(n) =0$, which we will refer to as {\em decay functions}.

For a given sequence $(\Lambda_n)$, increasing to $\Gamma$, and a decay function $f$, we define
\be
\Vert A \Vert_f = \Vert A \Vert + \sup_n  f(n)^{-1} \inf \{ \Vert A - A_n\Vert \mid A_n \in \cA_{\Lambda_n}\},
\label{norm_f}\ee
and
$$
\cA^f_\Gamma = \{ A \in \cA_\Gamma \mid \Vert A\Vert_f < \infty\}.
$$
It is not hard to show that $\cA^f_\Gamma$ is a Banach $^*$-algebra with norm $\Vert \cdot\Vert_f$ \cite{moon:2020}.

A concrete sequence of local approximations of any $A\in\cA_\Gamma$ is obtained by using the conditional expectations 
$\Pi_\Lambda$ determined by the tracial state, $\rho$, on $\cA_\Gamma$,
$$
\Pi_\Lambda = {\rm id}_{\cA_\Lambda}\otimes \rho\restriction_{\cA_{\Gamma\setminus\Lambda}}.
$$
The local approximations given by $\Pi_{\Lambda_n} (A)$ are not necessarily optimal but the error is always 
bounded by twice the optimal one: for any $A_n\in\cA_{\Lambda_n}$ we have $\Pi_{\Lambda_n}(A_n)=A_n$ and,
hence
$$
\inf \{ \Vert A - B\Vert \mid B \in \cA_{\Lambda_n}\}\leq\Vert A - \Pi_{\Lambda_n}(A)\Vert \leq \Vert A - A_n \Vert + \Vert \Pi_{\Lambda_n}(A_n-A) \Vert
\leq 2 \Vert A - A_n\Vert,
$$
where, for the last step, we used $\Vert \Pi_{\Lambda_n}\Vert = 1$. By replacing the $\inf$ in \eq{norm_f} 
by $\Vert A - \Pi_{\Lambda_n}(A)\Vert$ one obtains a more explicit equivalent norm.

For lattice systems with an infinite-dimensional Hilbert space at each site, such as oscillator lattices, a normalized trace does not exist, 
yet the same type of estimates can be obtained \cite{nachtergaele:2013}. 

{\em Lieb-Robinson bounds} express locality of observables by a commutator estimate, which is again equivalent 
to the error of local approximations up to a factor of $2$:
$$
\Vert A - \Pi_\Lambda(A)\Vert\leq\sup_{B\in\cA_{\Gamma\setminus\Lambda},\Vert B\Vert =1}\Vert [A,B]\Vert\leq 2 \Vert A - \Pi_\Lambda(A)\Vert.
$$
This relation also shows that $\cA_\Gamma$ for any countably infinite $\Gamma$, is asymptotically abelian in the sense that
$$
\lim_n \sup_{B\in\cA_{\Gamma\setminus\Lambda_n},\Vert B\Vert =1} \Vert [A,B]\Vert=0,
$$
for all $A\in\cA_\Gamma$ and increasing sequences $\Lambda_n \uparrow \Gamma$.

A common choice for the finite volumes $\Lambda_n$ is balls of radius $n$ centered at some $x\in \Gamma$.  For many families of decay functions that have been used in the literature so far, the choice of $x$ is not important as the norms defined with different $x$ are all equivalent. For some purposes, it is better to consider a Fr\'echet space of observables with a family of (semi-)norms \cite{kapustin:2022}.

\subsection{Interactions with sufficient decay}

A quantum spin model is typically defined in terms of an interaction: a map $\Phi$ defined on the finite subsets of $\Gamma$ 
such that $\Phi(X)=\Phi(X)^*\in\A_X$, for all finite $X\subset\Gamma$. For each finite $\Lambda\subset \Gamma$,
a local Hamiltonian is defined as follows:
$$
H_\Lambda = \sum_{X\subset \Lambda} \Phi(X),
$$
which in turn defines the finite-volume Heisenberg dynamics:
$$
\tau_t^\Lambda(A)= U_\Lambda(t)^* A U_\Lambda(t), \mbox{ with } U_\Lambda = e^{-itH_\Lambda}.
$$
$\{\tau_t^\Lambda\}_{t\in\Rl}$ is a group of $^*$-automorphisms of $\cA_\Gamma$ that leaves $\cA_\Lambda$ invariant.

If the interaction depends on time, $t\mapsto \Phi(X,t) \in \A_X$, bounded and measurable, then
\begin{align}
\frac{d}{dt} U_\Lambda(t,s) &= -i H_\Lambda(t) U_\Lambda(t,s)\notag\\
U_\Lambda(s,s) &= \idty,\notag
\end{align}
defines co-cycles of unitaries $U_\Lambda(t,s) $ and automorphisms 
$$
\tau_{t,s}^\Lambda(A) = U_\Lambda(t,s)^* A U_\Lambda(t,s).
$$
In applications we often make the further assumption that the $t$-dependence is piecewise differentiable
and $t$ may take values in a finite or infinite interval $I\subset \Rl$.

If for all $x\in\Gamma$, $\sum_{X, x\in X} \Vert \Phi(X)\Vert < \infty$, we can define a derivation $\delta: \A_{\rm loc}\to\A_\Gamma$, by
$$
\delta(A) = \sum_{X, X\cap Y\neq\emptyset} [\Phi(X), A] , \quad A\in \A_Y, \mbox{ finite } Y\subset\Gamma.
$$
Under conditions of sufficient decay, $\delta$ can be exponentiated to define $\tau^\Phi_t = e^{it\delta}$, a dynamics on $\A_\Gamma$.
In practice, one finds conditions on $\Phi$ such that $\tau_t^\Lambda(A)$ converges for all $A\in\A_{\rm loc}$ as $\Lambda\uparrow \Gamma$,
uniformly in $t$ in compact intervals.

A convenient way to state such a decay condition uses the notion of $F$-function, defined as follows.

\begin{enumerate}
\item $F:[0,\infty)\to(0,\infty)$, non-decreasing;
\item $\Vert F\Vert_1 := \sup_{y\in\Gamma} \sum_{x\in\Gamma} F(d(x,y))< \infty$;
\item $\sum_{z\in\Gamma} F(d(x,z)) F(d(z,y)) \leq C_F F(d(x,y|))$ for all $x,y\in\Gamma$ and a suitable constant $C_F$.
\end{enumerate}

Then, for an interaction $\Phi$ as above, possibly time-dependent, we define
$$
\Vert \Phi\Vert_F (t) = \sup_{x,y\in\Gamma} \frac{1}{F(d(x,y))} \sum_{X, x,y\in X}\Vert \Phi(X,t)\Vert, t\in I.
$$
Given $F$ and $I$ we can then define the Banach space of interactions continuous on $I$:
$$
\cB_F(I) = \{ \Phi(\cdot,\cdot) \mid \Phi(X,\cdot) \mbox{ continuous  for all finite } X\subset \Gamma, \mbox{ and }\Vert \Phi(t)\Vert_F \mbox{ is locally bounded}\}.
$$

For $\Gamma=\Ir^\nu$, it is easy to check that 
$$
F(r) =\frac{1}{(1+r)^{\nu+\epsilon}}, \mbox{ for any }\epsilon >0
$$
is an $F$-function. Furthermore $F$ of the following form are $F$-functions:
$$
F(r) = e^{-a g(r)} F_0(r),
$$
where $F_0$ is an $F$-function and $g:[0,\infty)\to [0,\infty)$ is sub-additive.
Useful choices for $g$ are $g(r) = r^\theta, \theta\in (0,1]$,  and 
$g(r) = r/(\log (1+r)^2)$. The first characterizes interactions with exponential and stretched exponential decay and the second applies to the typical case of interactions defining the quasi-adiabatic evolution (see \eq{special_weight}). The spaces of interactions defined by $F$-functions 
with stretched exponential pre-factor will also be denoted by $\cB_{a,\theta}$.

If for some $F$-function $F$, $\Vert \Phi\Vert_F (t)$ is locally bounded, it follows that the limit $\tau^\Phi_t(A)$ defines a strongly continuous
one-parameter group of automorphisms of $\A_\Gamma$, $\tau^\Phi_t = e^{it\overline{\delta}}$. In this situation,  $\overline{\delta}$ is the 
closure of $\delta$, i.e., $\cA_{\rm loc}$ is a core for $\overline{\delta}$.

In general, there are many interactions $\Phi$ that generate the same dynamics, and this freedom is useful.
For example, if $(\Gamma,d)$ is a Delone subset of $\Rl^\nu$, one can assume that $\Phi$ is supported on balls
$b_x(n)\subset \Gamma$ and express decay by a condition of the form
$$
\Vert\Phi(b_x(n))\Vert \leq \Vert \Phi\Vert_f f(n), \mbox{ for all } x\in\Gamma.
$$
A procedure to do this is described in \cite[Section 2.5]{nachtergaele:2021}. The generator $\delta$ can then be given on $\cA_g$, for some $g$, as
$$
\delta(A) = \sum_x [\Phi_x,A], \quad \Phi_x = \sum_{n\geq 1} \Phi(b_x(n)).
$$
For suitable $f$ and $g$, there is $h$ for which $\delta(A)\in \cA_h$.
A concrete example of practical use is the following. If $f(n) = g(n) = e^{-a n^\theta}$, for $a>0$, and $\theta \in (0,1]$, one can take
$h(n) = e^{-a' n^\theta}$, with $a'< a$.

\subsection{Lieb-Robinson bounds}

For interactions $\Phi \in \cB_F(\Rl)$, we have the following theorem.

\begin{thm}[Lieb-Robinson Bound]\label{thm:LRB}
Let $\Phi\in\cB_F(I)$, and $X,Y$ finite subsets of $\Gamma$ with $X\cap Y = \emptyset$. Then, for all $\Lambda\subset\Gamma$, 
$A\in \cA_X, B\in\cA_Y$, and $t< s\in I$, we have
\be
\left\Vert [\tau^{\Phi,\Lambda}_{t,s}(A),B]\right\Vert \leq C_F^{-1} 2\Vert A\Vert \Vert B\Vert \left(e^{2\int_s^t \Vert\Phi(r)\Vert_F dr} -1\right) 
\sum_{x\in X, y\in Y} F(d(x,y)).
\label{LRB}\ee
\end{thm}

To prove this theorem one proves it first for finite $\Lambda$. The finite-volume result can then be used to prove the existence of
the thermodynamic of the dynamics, which then also satisfies the bound stated in the theorem. It is worth noting that this result
for the thermodynamic limit of the dynamics for quantum spin systems goes beyond the traditional approach \cite{bratteli:1997,simon:1993} 
if one considers interactions in which all many-body terms may be non-vanishing. This is important for proving the locality properties of the 
quasi-adiabatic evolution (also known as the spectral flow). See \cite{nachtergaele:2019} for the details.

For time-independent interactions that decay exponentially, that is $\Phi\in \cB_F$, with $F(r) = e^{-ar}F_0$ for another $F$-function $F_0$,
one easily recovers the original form of the Lieb-Robinson bounds. It suffices to estimate the sum in the right-hand side of \eq{LRB} as follows:
$$
\sum_{x\in X, y\in Y} F(d(x,y)) \leq e^{-a d(X,Y)} \sum_{x\in X, y\in Y} F_0(d(x,y)) \leq \Vert F_0\Vert_1 \min( |X|, |Y|) e^{-a d(X,Y)}.
$$
Using the time-independence of $\Vert \Phi\Vert_F$, one then obtains immediately the usual exponential form:
\begin{align}
\left\Vert [\tau^{\Phi}_{t}(A),B]\right\Vert &\leq C_F^{-1} 2\Vert A\Vert \Vert B\Vert \left(e^{2|t|\Vert\Phi\Vert_F} -1\right)  
\Vert F_0\Vert_1 \min( |X|, |Y|) e^{-a d(X,Y)}\notag\\
&\leq  2\Vert A\Vert \Vert B\Vert C_F^{-1}  \Vert F_0\Vert_1  \min( |X|, |Y|) 
\left(e^{2|t|\Vert\Phi\Vert_F} -1\right) e^{-a d(X,Y)}\notag\\
&\leq C \Vert A\Vert \Vert B\Vert\min( |X|, |Y|) e^{a (v_{LR}|t|-d(X,Y))},\label{expLRB}
\end{align}
with 
$$
v_{LR} = 2 a^{-1} \Vert \Phi\Vert_F , \quad C= 2C_F^{-1} \Vert F_0\Vert_1.
$$

I already mentioned the applications of Lieb-Robinson bounds to proving exponential clustering of gapped ground states and generalizations 
of the Lieb-Schultz-Mattis theorem to higher dimensions \cite{hastings:2006,nachtergaele:2006a,hastings:2004,nachtergaele:2007}. In a few years after 
Hastings' rediscovery of Lieb-Robinson bounds, numerous applications appeared at an accelerating rate, in particular in quantum information theory and its applications to quantum many-body systems. In many cases they allow properties valid for finite systems to be extended to the thermodynamic limit by proving that they hold `uniformly in the volume'. A basic example of this type is the following theorem, which provides a continuity estimate for the dependence of the 
infinite-system dynamics on the interaction $\Phi$.

\begin{thm}[\cite{nachtergaele:2019}]
Let $\Phi, \Psi \in \cB_F(I)$, $X\subset \Gamma$ finite, and $A\in \cA_X$. Then, for all $s\leq t\in I$,
$$
\Vert \tau_{t,s}^\Phi (A) - \tau_{t,s}^\Psi(A)\Vert \leq 2 \Vert F\Vert_1 C_F^{-1} |X| \Vert A \Vert 
|t-s| e^{2|t-s| \min (\Vert \Phi\Vert_F, \Vert \Psi\Vert_F)} \Vert \Phi - \Psi\Vert_F.
$$
\end{thm}

The Lieb-Robinson bound of Theorem \ref{thm:LRB} applies to interactions that decay as a power law with sufficiently large exponent. This suffices for proving the existence of the thermodynamic limit of the dynamics and some other applications but in other situations one wants something better. For example, for exponentially decaying interactions the bound \eq{expLRB} shows the existence of a `light cone' that contains the essential support of time-evolved local observables.
To obtain an analogue of such a light cone for systems with power law interactions, a great deal of effort has been dedicated to improving the estimates in Lieb-Robinson bounds for such cases \cite{richerme:2014,matsuta:2017,else:2020,kuwahara:2020,tran:2021}. Another generalization with interesting applications concerns situations where $X$ or $Y$, or both, are infinite \cite{cha:2020,ogata:2021b}.

\section{Automorphic equivalence of gapped phases}

Gapped ground state phases are represented by open regions in some space of models (interactions). Much attention has been given in recent years to the topological classification of these phases. Instead of the usual notion of order parameters to characterize an ordered phase in statistical mechanics, one uses topological invariants (Chern numbers, various indices) to classify the topological structure of interest. Regardless of whether one is interested specifically in topological order or more generally in gapped ground state phases including those characterized by so called Landau order, one needs a mathematical description of which models and their ground state(s) belong to the same phase. 

Concretely, we would like to define an equivalence relation and identify phases with equivalence classes. Two seemingly different approaches have emerged in the literature. In the first, the emphasis is on parameter dependent model Hamiltonians. Two models are defined to represent the same gapped ground state phase if there is an interpolating curve of model Hamiltonians (interactions) along which the ground state gap satisfies a common positive lower bound. In the second, the focus is on the ground states themselves. Two states (which may or may not be assumed ground states of gapped Hamiltonians) are defined to belong to the same phase if there is a continuous family of sufficiently quasi-local automorphisms interpolating between them. Before investigating the relation between these two points of view, I give a precise formulation of each.

For concreteness, we will consider the class of interactions $\Phi\in \cB_F$, for an $F$-function $F$ on $(\Gamma, d)$, of stretched exponential form. Denote by $\delta^\Phi$ the corresponding generator of the Heisenberg dynamics on $\cA_\Gamma$. Furthermore, we restrict to the case of a finite number of mutually disjoint pure ground states $\cS^\Phi=\{\omega_1,\ldots, \omega_n\}$. We say that the model defined by $\Phi$ has {\em gapped ground states with gap $\gamma$}, $\gamma >0$, if for all $i=1,\ldots,n$, we have
\be
\omega_i (A^* \delta^\Phi(A)) \geq \gamma \omega(A^* A), \mbox{ for all } A\in\cA^{\rm loc}, \mbox{ with } \omega_i(A) =0.
\label{gapped_gs}\ee
This is equivalent to saying that the GNS Hamiltonian for the system in the ground state $\omega_i$ is a non-negative self-adjoint operator with a one-dimensional kernel and a spectral gap above zero of size $\geq \gamma$. Such ground states have been called {locally unique gapped ground states} by Tasaki 
 \cite{tasaki:2022}.

\subsection{Equivalence of interactions}

Suppose $\Phi_0$ and $\Phi_1$ are two interactions in the class $\cB_{a,\theta}$, with ground state sets $\cS^{\Phi_0}$ and  $\cS^{\Phi_1}$, respectively.

\begin{defn}\label{def:equivalent_interactions}\cite{chen:2010,nachtergaele:2019}
The interactions $\Phi_0$ and $\Phi_1$ belong to the same phase if there exists a differentiable curve of interactions $[0,1]\ni s \mapsto \Phi(s)$ such that the following hold:
\begin{enumerate}
\item $\Phi(0) = \Phi_0, \Phi(1) = \Phi_1$;
\item There exists a constant $\gamma'>0$, such that for all $s\in [0,1]$ $\Phi(s)$ has gapped ground states with gap $\gamma' >0$.
\item There exist $a'>0,\theta' \in (0,1]$, such that $\Phi (\cdot) \in \cB^1_{a',\theta'}([0,1])$, defined as the Banach space of interactions
for which, with $F(r)=e^{-a' r^{\theta'}}F_0(r)$,
$$
\sup_{x,y\in\Gamma} \frac{1}{F(d(x,y))} \sum_{{\rm finite} X: x,y\in X} \left( \Vert\Phi(X,s)\Vert + |X|  \Vert \Phi^\prime(X,s)\Vert\right)
$$
is bounded by a bounded measurable function of $s$.
\end{enumerate}
\end{defn} 

It is easy to show that `belonging to the same phase' defines an equivalence relation on a set of interactions.

\subsection{Equivalence of states}

Suppose $\cS_0$ and $\cS_1$ are two finite sets of pure states of $\cA_\Gamma$. 

\begin{defn}\label{def:equivalent_states}\cite{chen:2010,bachmann:2012}
The sets of states $\cS_0$ and $\cS_1$ are automorphically equivalent (in the stretched exponential locality class) if there exists a continuous
curve of interactions $[0,1]\ni s \mapsto \Psi(s)$ such that the following hold:
\begin{enumerate}
\item There exist $a'>0,\theta' \in (0,1]$, such that for all $s\in[0,1]$, $\Psi (s) \in \cB_{a',\theta'}$;
\item $[0,1]\ni s \mapsto \Psi(s)$ is continuous in the norm of $\cB_{a',\theta'}([0,1])$;
\item The family of automorphisms $\alpha_{s,0}$ generated by $\Psi(s)$ satisfies
$$
\cS_1 = \{ \omega\circ \alpha_{1,0}\mid \omega \in \cS_0\}.
$$
\end{enumerate}
\end{defn} 
 
It is also straightforward to show that `automorphic equivalence' defines an equivalence relation on sets of states.

\subsection{Equivalent equivalences}

We now show that the equivalence relations defined Definition \ref{def:equivalent_interactions} on sets of interactions on the one hand and 
Definition \ref{def:equivalent_states} on sets of states on the other hand are compatible in both directions.

\begin{thm}[From equivalence of states to equivalent interactions]
Let $\cS_0$ be a set of mutually disjoint pure ground states with a gap bounded below by $\gamma>0$
for the dynamics with generator $\delta_0$ defined by an interaction $\Phi_0\in\cB_{a,\theta}$, for some $a>0, \theta\in (0,1]$.
If a set of states $\cS_1$ is automorphic equivalent to $\cS_0$ in the stretched exponential locality class,
then there exists a differentiable curve of interactions of class $\cB^1_{a',\theta'}([0,1])$, for some $a'>0, \theta'\in (0,1]$, $\Phi(s), s\in [0,1]$, with $\Phi(0) =\Phi_0$, and such that  $\cS_1$ are gapped ground states with gap bounded below by $\gamma$ for the dynamics generated by $\Phi(1)$.
\end{thm}

\begin{proof}
Since each $\omega_0\in\cS_0$ is a gapped ground state with gap $\geq \gamma$ for the dynamics generated by $\delta_0$ defined in terms of an interaction
$\Phi(0)$, we have
$$
\delta_0(A) = \sum_{X\subset\Gamma, X\cap \supp(A)\neq \emptyset} [\Phi_0(X), A], \mbox{ for all } A\in\cA^{\rm loc},
$$
and
$$
\omega_i (A^* \delta_0(A)) \geq \gamma \omega_i(A^* A), \mbox{ for all } A\in\cA^{\rm loc}, \mbox{ with } \omega_i(A) =0.
$$
Let  $\alpha_s, s\in [0,1]$, be the curve of automorphisms implied by the automorphic equivalence of $\cS_1$ with $\cS_0$, and define
$$
\delta_s = \alpha_s^{-1} \circ \delta_0 \circ \alpha_s, s\in [0,1].
$$
It is straightforward to check that the action of the generator $\delta_s$ on local observables $A$ can be given as
$$
\delta_s(A) =  \sum_{X\subset\Gamma, X\cap \supp(A)\neq \emptyset} [ \alpha_s^{-1} (\Phi_0(X)), A].
$$
The map $X\mapsto \tilde\Phi(X,s) :=\alpha_s^{-1} (\Phi_0(X))$ is not an interaction in the usual sense because  $\tilde\Phi(X,s)$ are generally not local.
It is shown in \cite[Section VI.E.2]{nachtergaele:2019} how to construct a proper interaction of the stretched exponential class, $\Phi(s)$, that generates the 
correct infinite system dynamics, meaning  $\tau^{\Phi(s)}_t =  \alpha_s^{-1} \circ \tau^{\Phi_0}_t \circ \alpha_s, s\in [0,1], t\in\Rl$. Of course, $\Phi(s)$ 
also leads to the same generator $\delta_s$.

Then, using $\cA_{\rm loc} \subset \cA_f \subset \dom \delta_s$, and $\alpha_s^{\pm 1} (\cA_f)\subset \cA_f$, for suitable $f$, we can
verify that the gap remains bounded below by $\gamma$ along the curve (in fact, the gaps for each pure state are constant as a function of $s$).
To see this, take $A\in\cA^{\rm loc}$, such that $\omega_s(A) =0$. Then, by the assumptions, $\alpha_s(A)\in \dom \delta_0$ and  $
\omega_0(\alpha_s(A)) =0$, we have
\begin{align}
\omega_s(A^* \delta_s(A)) &= \omega_0(\alpha_s(A^*) \alpha_s(\delta_s(A))) \notag\\
&=\omega_0(\alpha_s(A^*) \delta_0(\alpha_s(A))) \notag\\
&\geq \gamma \omega_0(\alpha_s(A^*) \alpha_s(A)) \notag\\
&= \gamma \omega_s(A^*A).\notag
\end{align}
This proves that $\omega_1$ is a ground state for $\delta_1$ and with the same gap as $\omega_0$ for $\delta_0$, and
$\omega_s, s\in [0,1]$ are gapped ground states of the interpolating family $\Phi^{(s)}$. Since $\omega_0\in \cS_0$ was arbitrary, this conclude the proof.
\end{proof}

\begin{thm}[From equivalent interactions to automorphic equivalence]
Suppose $s\mapsto \Phi(s)$ is a differentiable curve of interactions of class $\cB^1_{a,\theta}([0,1])$, such that there exists $\gamma>0$ 
and sets of mutually disjoint pure gapped ground states $\cS_s$, $s\in [0,1]$, with gap bounded below by $\gamma$. 
Then, there exists a strongly continuous curve of automorphisms $\alpha_{s}$ of class $\cB_{a',\theta'}([0,1])$, such that 
$$
\cS_s = \{ \omega\circ \alpha_{s,0}\mid \omega \in \cS_0\}.
$$
\end{thm}

\begin{proof}
For any differentiable curve of interactions one can define the Hastings generator of what is often referred to as the quasi-adiabatic evolution\cite{hastings:2004,hastings:2005,bachmann:2012,nachtergaele:2019}. It is defined by transforming the following
`non-local interaction' into a standard interaction:
$$
\tilde\Psi(X,s)=
\int_{-\infty}^\infty w_\xi(t) \int_0^t \tau_u^{(s)} \left(\frac{d}{ds} \Phi(X,s)\right)du\, dt.
$$
Here, there are in principle many possible choices for the function $w_\xi(t)$. A specific choice is the $L^1$-normalized positive function 
discussed in detail in \cite[Section VI.B]{nachtergaele:2019}. It satisfies
\be
w_\xi(t) \leq c \xi |t| e^{-\eta\frac{\xi |t|}{(\ln (\xi |t|))^2}}
\label{special_weight}\ee
for some numerical constants $c >0$ and $\eta > 2/7$, and $\xi$ is a parameter that needs to satisfy $\xi\in(0,\gamma)$. A larger value of 
$\xi$ generally leads to better locality properties for the automorphisms $\alpha_s$ generated by the $\Psi(s)$. See \cite[Section VI]{nachtergaele:2019} for an in depth discussion. The upshot is that the curve $\alpha_s$ is generated by a `time-dependent' interaction $\Psi(s) $ of a stretched exponential decay class:
$$
\frac{d}{ds} \alpha_s (A) = i\alpha_s(\delta^{\Psi(s)}(A)), A \in \cA^{\rm loc}.
$$
The crucial property of Hastings' quasi-adiabatic evolution $\alpha_s$ is then that for any $\omega_0\in\cS_0$, $\omega_s := \omega_0 \circ \alpha_s \in \cS_s$, which proves the automorphic equivalence of the sets of states $\cS_0$ and $\cS_1$.
\end{proof}

\section{Remarks}

There are a number of topics closely related to the notion of gapped ground state phase discussed in this paper that I have not discussed. The first obvious question one would raise in connection to the gapped phases is their stability. Are the interactions that lead to gapped ground states part of an open set of interactions for which the gap persists? This type of stability property has been the goal of much research over the past several decades. Some recent approaches include \cite{yarotsky:2006,de-roeck:2019,frohlich:2020}. A very successful approach that uses Lieb-Robinson bounds at its core is the Bravyi-Hastings-Michalakis strategy of \cite{bravyi:2010,bravyi:2011,michalakis:2013} further developed in \cite{nachtergaele:2021,nachtergaele:2021c}.

For concreteness and brevity, I focused the discussion on quantum spin systems, but note that Lieb-Robinson bounds and automorphic equivalence can be proved and applied in the same way for lattice fermion systems \cite{hastings:2006,bru:2017,nachtergaele:2018,moon:2018,bourne:2021,ogata:2021c}.

\section*{Acknowledgements}

As so many others, I continue to draw inspiration from Elliott Lieb's pathbreaking work on quantum lattice systems. This acknowledgement is too short to do justice to the profound  influence Elliott has had on my way of doing mathematical physics. Of the many collaborators with whom I had the pleasure of working on the topics discussed in this paper, I want to acknowledge two in particular: Robert Sims (U Arizona) and Amanda Young (TU Munich \& MCQST). 

\section*{Funding}
This paper is based upon work supported by the National Science Foundation under grant DMS--2108390.
 

\end{document}